\documentclass{article}

\usepackage{graphicx}
\usepackage{amsfonts}
\usepackage{amsmath}
\usepackage{amsthm}
\usepackage[ruled]{algorithm2e}

\def\F{\ensuremath{\mathbb{F}}}   
\def\S{\ensuremath{\mathcal{S}}}  
\def\A{\ensuremath{\mathcal{A}}}  
\def\H{\ensuremath{\mathcal{H}}} 
\def\R{\ensuremath{\rho}}              
\def\bound{\ensuremath{T}}           
\def\C{\ensuremath{\mathcal{C}}}  
\def\kpdual{\ensuremath{r'}}          
\def\puncSet{\ensuremath{\mathcal{P}}} 
\def\O{\ensuremath{\mathcal{O}}}          

\def\Y{\ensuremath{\mathcal{Y}}}

\def\E{\ensuremath{\mathcal{E}}}

\newcommand\vect[1]{\ensuremath{\textbf{#1}}}
\newcommand\parenthesis[1]{\ensuremath{\left( #1 \right)}}

\newcommand\ceil[1]{\ensuremath{\left\lceil #1 \right\rceil}}
\renewcommand\brace[1]{\ensuremath{\left\lbrace #1 \right\rbrace}}

\newtheorem{theorem}{Theorem}
\newtheorem{proposition}{Proposition}
\newtheorem{example}{Example}

\DeclareMathOperator{\Emb}{Emb}
\DeclareMathOperator{\Ext}{Ext}
\DeclareMathOperator{\dec}{dec}
\DeclareMathOperator{\cl}{cl}
\DeclareMathOperator{\BCH}{BCH}

\begin{document}

\title{Improving success probability and embedding efficiency in code
  based steganography}

\author{
  Morgan Barbier\thanks{University of Caen - GREYC
    \texttt{morgan.barbier@unicaen.fr}}
  \and
  Carlos Munuera\footnote{University of Valladolid - Department of
    Applied Mathematics \texttt{cmunuera@arq.uva.es}}
  \thanks{This work was
    supported by Junta de Castilla y Le\'on under grant VA065A07 and by
Spanish Ministry for Science and Technology under grants
MTM2007-66842-C02-01 and  MTM 2007-64704.}}
\sloppypar
\maketitle

\begin{abstract}
For stegoschemes arising from error correcting codes, embedding depends on a decoding
map for the corresponding code. As decoding maps are usually not
complete, embedding can fail.
We propose a method to ensure or increase the probability of embedding
success for these stegoschemes. This method is based on puncturing
codes.  We show how the use of punctured codes may also increase the
embedding efficiency of the obtained stegoschemes.
\end{abstract}


\section{Introduction}

Steganography is the art of transmitting information in secret, so
that even the existence of communication is hidden. It is realized by
embedding the messages to be protected into innocuous cover objects
such as digital images. In order to minimize the possibility of being
detected by third parties, the number of embedding changes in the
cover must be small enough. Consequently, to obtain a high payload,
steganographers should design stegosystems able to embed as much
information as possible per cover change. In other words, we seek for
embedding methods giving high embedding efficiency.\\

Given a cover image and a secret message, we first select the
placement and intensity of allowed embedding changes in the cover. This
is done by means of a {\em selection rule}. After this step the cover
is transformed into a finite sequence $x_1,\dots,x_n$ of symbols from
an alphabet $\mathcal{A}$ (usually $\mathcal{A}=\mathbb{F}_2$). We
refer to this sequence as the {\em cover vector}
$\vect{x}=(x_1,\dots,x_n)$. The secret message will be also a vector
$\vect{m}=(m_1,\dots,m_r)\in \mathcal{A}^r$. The algorithms used for
embedding the information $\vect{m}$ into
$\vect{x}$ and later recovering $\vect{m}$ from $\Emb(\vect{x},\vect{m})$)form the {\em
  stegoscheme} associated to the stegosystem.\\

Crandall first noted that error correcting codes can be used to
construct stegoschemes of high embedding efficiency. From this
discovering many stegoschemes have been proposed using different types
of codes. Relations between codes and stegoschemes will be treated in
more detail in Section \ref{codsteg}. In general, in code based
steganography (also called {\em matrix embedding}), given a code $\C$
(or a family of such codes, as explained below), the embedding is
realized by using a decoding map of $\C$. It turns out that every
decoding map of $\C$ provides a realization of it as a
stegoscheme. This nice idea encountered in practice a serious problem:
the absence of effective complete decoding methods. Indeed, for nearly
all currently known decoding algorithms, most errors are impossible to
decode, which means that when using these algorithms, the majority of
times the embedding process fails. In contrast, complete decoding
algorithms, for which this problem does not occur, are computationally
infeasible, so their practical applications are reduced  to either perfect
codes or codes of small length, which offer low embedding efficiency.\\

In this paper we propose a novel method to ensure embedding success, in
principle for using any code. It is based on puncturing the original
code as many times as necessary to get a new code whose covering
radius equals to the correction capability of the original one. As we
shall see, this method may also improve the embedding efficiency so
that it guaranties embedding success and may give high embedding
efficiency.\\

The organization of the paper is as follows. In Section \ref{codsteg}
we recall the connection between coding theory and steganography, as
well as we point the above mentioned drawback of stegoschemes obtained
by this method, concerning decoding maps. Our method is exposed in
detail in Section 3. Finally Section 4 deals to the parameters of the
new stegoschemes. Some numerical experiments concerning BCH codes are
reported.

\section{From coding theory to steganography}\label{codsteg}

\subsection{Stegoschemes and codes}

Let $\A$ be a finite alphabet with $q$ elements and let $n\ge r$ be
two positive integers. The purpose of a stegoscheme $\mathcal{S}$ is
to embed a message $\vect{m} \in \A^r$ into a cover vector $\vect{x}
\in \A^n$, making as few changes as possible in $\vect{x}$, and later
extract the information hidden in the modified vector. Formally, a
$(n,r)$-stegoscheme $\mathcal{S}$ is defined as a couple of
functions $\mathcal{S}=(\Emb,\Ext)$
$$
\begin{array}{rclcrcl}
  \Emb: \A^n \times \A^r & \longrightarrow & \A^n & \mbox{ and } &
  \Ext: \A^n & \longrightarrow & \A^r,\\
\end{array}
$$
such that $\Ext(\Emb(\vect{x},\vect{m})) = \vect{m}$, for all
$(\vect{x},\vect{m}) \in \A^n \times \A^r$.
This condition guarantees that we always retrieve the right message
$\vect{m}$ from the stego vector $\Emb(\vect{x},\vect{m})$. The
stegoscheme is called {\em proper} if
$d(\vect{x},\Emb(\vect{x},\vect{m}))\le d(\vect{x},\vect{v})$ for all
$\vect{v}$ such that $\Ext(\vect{v})=\vect{m}$, where $d$ stands for
the Hamming distance. This means that the number of embedding changes
is the minimum possible allowed by the extracting map.\\

The first stegoscheme based on coding theory was proposed by Crandall
in 1998, by  using the family of binary Hamming codes,
see~\cite{Crandall}. Let $m$ be a positive integer and  let $H$ be  a
parity check matrix of the binary Hamming code of length
$n=2^m-1$. The embedding and extracting functions are as follows
$$
\begin{array}{rcl}
 \Emb: \F_2^n \times \F_2^{m} & \longrightarrow & \F_2^{n}\\

(\vect{x},\vect{m}) & \longmapsto & \vect{x}- \cl(\vect{x}H^T-\vect{m}),\\
\end{array}
$$
$$
\begin{array}{rcl}
\Ext: \F_2^{n}  & \longrightarrow & \F_2^m\\
\vect{v} & \longmapsto &\vect{v}H^T,
\end{array}
$$
where $\cl(\vect{z})$ stands for a {\em coset leader} of $\vect{z}\in\F_2^m$.
That is an element of minimum Hamming weight among those
$\vect{v}\in\F_2^{n}$ verifying $\vect{v}H^t =\vect{z}$. The fact that
$\Ext(\Emb(\vect{x},\vect{m})) = \vect{m}$ is straightforward.\\

This stegoscheme is so efficient that Westfeld developed the
famous software F5 based on it~\cite{F5_algo}. The notion of {\em
efficiency}, which is the main subject of this paper, will be
detailed in Section~\ref{ssec:efficiency}. After Crandall's
discovery, other authors proposed  the same model of stegoscheme,
based on different types of codes:
BCH~\cite{syndromeCodingBCH,BCH_fastSyndrome},
Reed-Solomon~\cite{FontaineGaland}, trellis
codes~\cite{MinimizingImpact}, etc.\\

More generally, in~\cite{MunueraBarbier_Systematic} it is shown that a
$(n,r)$-stegoscheme is equivalent to a family
$\brace{(\C_\vect{m},\dec_\vect{m})}_{\vect{m}\in\A^r}$  where the
$\C_\vect{m}$'s  are nonempty disjoint codes in $\A^n$ (not necessarily linear),
$\dec_\vect{m}$ is a decoding map for $\C_\vect{m}$ and
$\cup\C_\vect{m}=\A^n$. Let us remember that a {\em decoding
  map} for a code $\C\subseteq\A^n$ is just a map $\dec:
\mathcal{X}\subseteq\A^n\longrightarrow \C$. If
$\mathcal{X}=\A^n$ then the decoding is said to be {\em
  complete} (as every received vector can be decoded). If $\dec$
verifies the additional property that
$d(\vect{x},\dec(\vect{x}))=d(\vect{x},\C) $ for all
$\vect{x}\in\mathcal{X}$, then $\dec$ is called {\em minimum
  distance decoding}. Given such a family, the corresponding
stegoscheme has embedding and extracting functions
$$
\Emb(\vect{x},\vect{m})=\dec_\vect{m}(\vect{x}) \; \mbox{ and } \;
\Ext( \vect{v})= \vect{m}    \mbox{ if $\vect{v}\in\C_{\vect{m}}$.}
$$
If $\dec$ is a minimum distance decoding, then the obtained
stegoscheme is proper. This method allows us to construct a large
amount of stegoschemes, to which we collectively refer as {\em
  code-based} stegoschemes. If  $\A$ is a field, $\A=\F_q$, and
$\C_\vect{0}$ is $[n,n-r]$ linear, it is natural to consider the
partition of
$\F_q^n$ given by the translates of $\C_\vect{0}$ (or equivalently,
the cosets of $\F_q^n/\C_\vect{0}$). Note that given a decoding map
$\dec_\vect{0}$ of $\C_\vect{0}$, the function
$\dec_{\vect{m}}(\vect{x})=\cl(\vect{m}) +
\dec_{\vect{0}}(\vect{x}-\cl(\vect{m}))$ is a decoding map for the
translate $\C_\vect{m}=\cl(\vect{m})+\C_\vect{0}$. By using the
systematic writing of $\C_{\vect{0}}$, say in the first $n-r$
positions, we can avoid the need of computing coset leaders in the embedding process
(unless they are required by
$\dec_{\vect{0}}$), as $\Emb(\vect{x},\vect{m}) =
(\vect{0},\vect{m})+\dec_{\vect{0}}(\vect{x}-(\vect{0},\vect{m}))$,
see~\cite{MunueraBarbier_Systematic}. The same idea can be applied if
$\C_\vect{0}$ is a group (nonlinear) code or a general systematic
code.\\

Therefore, given a linear code $\C$, each decoding map of $\C$
provides a {\em realization} of $\C$ as stegoscheme. Recall that there is a
universal decoding method for linear codes, the so-called {\em
  syndrome-leader decoding}, based on pre-computing all syndromes and
leaders. For example, Crandall's stegoscheme is obtained from Hamming
codes by using syndrome-leader decoding.

\subsection{Embedding Efficiency}
\label{ssec:efficiency}

The behavior of a stegoscheme, and subsequently the comparison of two
of them, is based on its parameters. Let $\S$ be a $(n,r)$-stegoscheme
which embeds a message of $\A^r$ in a vector of $\A^n$ with  $\bound$ 
modifications at most and $\tilde{\bound}$ modifications in average. The {\em
  relative payload, change rate} and {\em average change rate} of $\S$
are respectively defined as
  $$
  a = \frac{r}{n} \;  , \; 
  R=\frac{\bound}{n} \; \mbox{ and }   \; \tilde{R}=\frac{\tilde{\bound}}{n}.
  $$
The {\em embedding efficiency} and {\em average embedding efficiency}
of $\S$ are defined as 
$$
 e = \frac{r}{\bound} \mbox{ and }  \tilde{e} = \frac{r}{\tilde{\bound}}.
$$
Recall that when $\S$ is arising from a linear or systematic code
$\C$ (by the method above explained), then $r$ is the redundancy of $\C$, and $T, \tilde{T}$ are the
covering and average radii of $\C$, see for
example~\cite{MunueraBarbier_Systematic}.\\

Let $\S_1$ and $\S_2$ be two stegoschemes defined over the same 
alphabet $\A$. If they have the same relative payload, then both embed
as much information by using the same
quantity of cover-medium. Thus  $\S_1$ is better
than $\S_2$  if and only if  it produces less distortion in the cover,
that is, if and only if 
the embedding efficiency of $\S_1$ is greater than the embedding efficiency of
$\S_2$.  So we look for stegoschemes with the biggest embedding
efficiency for a fixed relative payload. There exists
an upper on the embedding efficiency for a fixed relative payload as follows.

\begin{theorem}
  Let $\S$ be a $q$-ary stegoscheme with relative payload $a$ and embedding
  efficiency $e$. Then
  $$
  e \le \frac{a}{\H_q^{-1}(a)},
  $$
  where $\H_q^{-1}$ is the inverse function of the $q$-ary entropy
  $\H_q(x) = x \log_q(q-1) - x \log_q(x) - (1-x) \log_q(1-x)$.
\end{theorem}
See \cite[Ch. 12, page 454]{book:WatermarkingSteganography}.
The same result holds asymptotically for
the average embedding efficiency.

\subsection{Realizations of codes as stegoschemes}
\label{ssec:realizations}

Let $\C$ be a  $[n,k]$ linear code over $\A=\F_q$. As seen in the
previous section, each decoding map $\dec$ of  $\C$ provides a
realization of $\C$ as a $(n,r)$ stegoscheme, $\S_{\dec}$. All of them
are different although they share the same theoretical parameters:
$r=n-k$ is the redundancy of $\C$, ${R}$ is the covering radius of
$\C$ and  $\tilde{R}$ is the average radius of
$\C$. This method of making
stegoschemes has, in practice, a serious difficulty: the dramatic
absence of computationally efficient complete decoding maps for a
given code. As already mentioned above, there is a universal complete decoding method for linear
codes, namely syndrome-leader decoding. For binary codes we can also use
gradient decoding, based on minimal codewords.  However, these two
methods require the use (computation and storage) of tables, usually
of a very large size, so they have mainly a theoretical interest,
being their practical applications restricted to the case of perfect
codes or codes of small length $n$. Note that stegoschemes of high
efficiency have large
$n$. Algebraic decoding methods (not based on using tables) may be
computationally efficient, but they are far to be complete; usually they
decode up to the packing radius $t=\lfloor(d-1)/2\rfloor$, where $d$
is the minimum distance, or even less. For example, Berlekamp-Massey
algorithm for decoding for BCH codes, decodes errors of weight up to
one half of the designed minimum distance. For BCH two-error
correcting codes (which are quasi-perfect) a complete decoding method
is available in \cite{Hartmann} by using small tables and these codes
have been proposed
as good candidates for stegoschemes, see \cite{BCH_fastSyndrome}. For
other codes most vectors cannot be decoded. Consequently, most vectors
of $\A^n$ cannot be used as a cover vectors for embedding. Remark that
when using codes for error-correction, error patterns of low weight
are more probable while for embedding purposes all vectors are, in
principle, equally likely as covers.

\begin{example}\label{exBCH}
Let $\C=\BCH_m(3)$ be the primitive binary triple error correcting BCH
code of length $n=2^m-1$. Suppose first we want to decode $\C$ by using syndrome-leader decoding. To that end we need a table containing all syndromes and leaders. Since the total number
of cosets is $2^{3m}=(n+1)^3$, the size of this table, in megabits, is the given in the following Table 1.

\begin{table}[h]
  $$
  \begin{array}{|c||r|}
    \hline
    m & \mbox{size in Mb} \\
    \hline
    \hline 
5 & 1.507\\  \hline  
6 & 21.234\\  \hline  
7 & 310.378\\  \hline  
8 & 4680.843\\  \hline  
9 & 72209.138\\  \hline  
10& 1130650.141\\ \hline
  \end{array}
  $$
  \caption{\label{tab:efficiencyBCH(3)}
  Size of a syndrome-leader table for the code $\BCH_m(3)$.}
\end{table}

It is clear that we need other decoding methods even for small values of $m$.
Let us consider the minimum distance decoding map $\dec$
obtained by means of Berlekamp-Massey algorithm. It corrects up to 3
errors. Since the covering radius of this code is $5$, see
\cite{Helleseth}, error patterns of weight $\le 5$ can occur. Let
$A_j$ denote the number of cosets having leaders of weight $j$. Since
every vector of weight $\le 3$ is a coset leader and there are
$2^{3m}=(n+1)^3$ cosets, then the number of cosets whose vectors
cannot be decoded by $\dec$ is
$$
A_4+A_5=(n+1)^3-\sum_{j=0}^3 
\left( \begin{array}{c} n \\ j \end{array} \right) 
=\frac{n(n+1)(5n+13)}{6}.
$$
Since the leading coefficient of the numerator is $5$, asymptotically
$5/6$ of error patterns cannot be decoded. Consequently, the
stegoscheme based on $\C_m$ and $\dec$,
successfully embeds a message into a random cover vector with
probability close to $1/6$.  Let us note, however, that it is known a
computationally efficient decoding algorithm for $\C$, see
\cite{triple}. In fact BCH(2) and  BCH(3) are the only codes  in the
family of BCH codes for which such a complete decoding is available.
\end{example}

Given a $(n,r)$ stegoscheme $\S=(\Emb,\Ext)$, the previous considerations allow us to define the {\em
  embedding probability} $p_\S$ of $\S$ as the probability that
$\Emb(\vect{x},\vect{m})$ can be computed over all possible pairs
$(\vect{x},\vect{m}) \in \A^n\times\A^r$ (considered as
equally likely). We define also the {\em embedding efficiency of $\S$
  relative to the embedding probability} $p_\S$ (resp.  {\em average
  embedding efficiency of $\S$ relative to the embedding probability})
as
$$
 e_{rel} = e . p_\S \mbox{ and }  \tilde{e}_{rel} = \tilde{e} .
 p_\S.
$$

\begin{example} (Example \ref{exBCH} continued).
Let $\S$ be the stegoscheme obtained from $\C=\BCH_m(3)$. The
parameters of $\S$ can be obtained from the parameters of $\C$. In
particular, $n=2^m-1$, $r=3m$ (the redundancy of $\C$) and $\bound=5$
(the covering radius of $\C$). In the same way, we can obtain the
average change rate $\tilde{\bound}$ of $\S$ as the average radius of
$\C$,
$$
\tilde{\bound}=2^{-3m}\sum_{j=0}^5 j A_j.
$$
Since the values of $A_4$ and $A_5$ are not known in general, we can not compute
exactly this number. However, we have the following bounds
\cite{triple}
$$
\frac{5n(5n+13)}{6} \le A_4 \le \frac{n(5n^2+10n-3)}{6},
$$
$$
\frac{4n(n+2)}{3} \le A_5 \le \frac{n(n-4)(5n+13)}{6}.
$$
The embedding probability can be deduced from our
computations in Example \ref{exBCH}, as
$$
p_\S=2^{-3m} \sum_{j=0}^3 
\left( \begin{array}{c} n \\ j \end{array} \right) 
\approx \frac{1}{6}.
$$
The following Table 2 collects the efficiency and relative efficiency of
these stegoschemes for some small values of $m$.

\begin{table}[h]
  $$
  \begin{array}{|c||c|c|c||c|c||c|c|c|}
    \hline
    m & n &  r & \tilde{\bound} & e & \tilde{e} & p_\S &  e_{rel}
    \rule{0mm}{4mm} & \tilde{e}_{rel}\\
    \hline
    \hline
    4 & 15  & 10 & 3.33 & 2   & 3.003 & 0.141 & 0.281 & 0.422\\ \hline
    5 & 31  & 15 & 4.28 & 3   & 3.504 & 0.152 & 0.457 & 0.532\\ \hline
    6 & 63  & 18 & 4.06 & 3.6 & 4.433 & 0.159 & 0.573 & 0.704\\ \hline
    7 & 127 & 21 & 3.85 & 4.2 & 5.454 & 0.162 & 0.684 & 0.883\\ \hline
    8 & 255 & 24 & 3.84 & 4.8 & 6.250 & 0.163 & 0.791 & 0.891\\ \hline
    9 & 511 & 27 & 3.84 & 5.4 & 7.031 & 0.166 & 0.896 & 1.167\\ \hline
    10 & 1023 & 30 & 3.84 & 6 & 7.812 & 0.166 & 0.996 & 1.296\\ \hline
  \end{array}
  $$
  \caption{\label{tab:efficiencyBCH(3)}
    Embedding efficiency of stegoschemes based on $\BCH_m(3)$ codes.}
\end{table}
\end{example}

\section{A method to ensure embedding success}

As we highlighted in the previous section, it is important to have a
stegoscheme with an embedding probability as close to one as
possible. In this section, we propose to modify the stegoscheme
arising from a given code in order to ensure success (or at least to
improve its probability).

\subsection{Puncturing codes}

Let $\C$ be a $[n,n-r,d]$ code of covering radius $\rho$. Let $\dec$
be a decoding map of $\C$, which can decode up to $t$ errors. The main
idea of our method is to imitate, in a sense, the behavior of perfect codes. 
To that end, we will puncture $\C$ as many times as necessary to
obtain a code $\C'$ with covering radius $\rho'=t$. Then we can use
$\dec$ to deduce a decoding map for $\C'$, as explained in the next
paragraph.  Let us remember that given a code $\C$, {\em puncturing}
$\C$ at a position $i$ is to delete the $i$-th coordinate in each
codeword. If $G$ is a generator matrix for $\C$, then a generator
matrix for the punctured code is obtained from $G$ by deleting column
$i$ (and dependent rows if necessary), see
\cite[Sect. 1.5.1]{book:HuffmanPless}.

\subsection{Decoding punctured codes}

Let $\C$ be a binary linear code and $\C'$ be the code obtained
from $\C$ by puncturing at a set of positions $\puncSet \subset
\brace{1,\dots, n}$. Let $\overline{\puncSet}= \brace{1,\dots,
  n}\setminus\puncSet$,  $\pi_{\overline{\puncSet}}$ be the projection
on the coordinates of $\overline{\puncSet}$ and $\dec$ a decoding map
for $\C$. The following algorithm provides a decoding map for
$\C'$.\\

\begin{algorithm2e}
  \SetKwInOut{Input}{Input}\SetKwInOut{Output}{Output}

  \Input{The punctured set $\puncSet$, the received vector $\vect{y}'\in
    \F_q^{n-|\puncSet|}$ and the decoding map $\dec$ of $\C$.}
  \Output{The list of all codewords $\vect{c}'$ of $\C'$ such that
    $d(\vect{c}',\vect{y}')\le t$.}

\BlankLine
\BlankLine
  
  $\mathcal{L} \gets \emptyset$;\\
  \hfill// for $q^{|\puncSet|}$ elements\\
  \ForEach{$\vect{y}\in
    \F_q^n$ such that
    $\pi_{\overline{\puncSet}}(\vect{y}) =\vect{y}'$}
  {
    $\mathcal{L} \gets \mathcal{L} \cup \dec(\vect{y})$;
  }
  \Return{$\brace{ \pi_{\overline{\puncSet}}(\vect{c}) : \vect{c} \in
      \mathcal{L}}$}
  \caption{\label{algo:decCPunct}
    Decoding algorithm for the punctured code $\C'$.}
\end{algorithm2e}

\begin{proposition}
  \label{prop:correctnessComplexity}
  The algorithm~\ref{algo:decCPunct} is correct and runs in
  $\O(q^{|\puncSet|} c_{\dec})$, where $c_{\dec}$ is the complexity of
  $\dec$ the decoding map of $\C$.
\end{proposition}
\begin{proof}
  Let $\vect{y}' \in \F_q^{n-|\puncSet|}$ and $\vect{c}'\in \C'$ be a codeword
  such that $d(\vect{y}',\vect{c}') \le t$. Then there exist $\vect{c}
  \in \C$ such that
  $\pi_{\overline{\puncSet}}(\vect{c})=\vect{c}'$. Let us to note that
  $\vect{y}\in \F_q^n$ such that $\pi_{\puncSet}(\vect{y}) =
  \pi_{\puncSet}(\vect{c})$ and $\pi_{\overline{\puncSet}}(\vect{y})
  =\vect{y}'$. Then $\vect{c}'$ is in the returned list.
  The statement of the complexity is obvious. 
\end{proof}

If $\dec$ corrects up to $t$ errors, the previous algorithm provides a
decoding map for $\C'$ correcting $t$ errors as well: as the returned
list is nonempty, simply take one of the vectors closest to
$\vect{y}'$.

Note that the statement about complexity in Proposition
\ref{prop:correctnessComplexity} has been computed considering the
worst case. In following, we propose a little improvement for the
average case.\\

Let $\vect{y}' \in \F_q^{n-|\puncSet|}$, $\vect{y}' \not \in \C'$, be
the vector to be decoded and let $\vect{c}' \in \C'$ be the closest
codeword of $\vect{y}'$. Then $d(\vect{c}', \vect{y}') = d(\C',
\vect{y}')$. Let $\vect{e}' \in \F_q^{n-|\puncSet|}$ be such that
$\vect{y}' = \vect{c}' + \vect{e}'$, hence $\mbox{wt}(\vect{e}') \le
\R' = t$. There exists a ``prefix'' vector  $\vect{p} \in
\F_q^{|\puncSet|}$ such that $\vect{c} = (\vect{p},\vect{c}') \in \C$,
with $\pi_{\puncSet}(\vect{c}) = \vect{p}$ and
$\pi_{\overline{\puncSet}}(\vect{c}) = \vect{c}'$. Finally, let
$\vect{y}, \vect{e} \in \F_q^n$ such that $\vect{y} = (\vect{p},
\vect{y}')= \vect{c} + \vect{e}$. As stated before, Algorithm 1 works
by considering all possible prefixes. However, since $\vect{e}=(
\vect{0}, \vect{e}')$, and hence $\mbox{wt}(\vect{e}) =
\mbox{wt}(\vect{e}')$, this algorithm can be slightly changed to avoid
unnecessary iterations. In fact, the minimum number of iterations
depends on $\mbox{wt}(\vect{e}')$. It is simple to see that the number
of iterations needed to decode is
$$
\ceil{\frac{q^{|\puncSet|}}{V_q(|\puncSet|, t-\mbox{wt}(\vect{e}'))}}.
$$
Furthermore, all this iterations should be parallelized to speed up the
algorithm. Of course, the value $\mbox{wt}(\vect{e}')$ is not known a
priori, but for each codeword $\vect{c}_i \in \mathcal{L}$, then
distance $d(\vect{c}_i, \vect{y})$ gives an upper bound to the distance
between $\vect{y}$ and the closest codeword $\vect{c}$.

\begin{example} (Continued from previous examples).
Consider the code BCH$_4(3)$ of parameters $[15,5]$. Here $t=3$ and,
as we can see in Table~\ref{tab:PC3BCH}, we get $|\puncSet| = 3$.
  \begin{table}
    \centering
    \begin{tabular}{c|c|c|c}
     (0, 0, 0) & (0, 0, 1) & (0, 1, 1) & (1, 1, 1)\\
      & (0, 1, 0) & (1, 0, 1) & \\
      & (1, 0, 0) & (1, 1, 0) & \\
    \end{tabular}
    \caption{\label{tab:vect3}
    All binary vectors of length 3 sorted by weights.}
  \end{table}

  Let $\C$ be the binary BCH$_4(3)$ code, $\puncSet = \brace{1, 2, 3}$
  be a punctured set and $\C'$ be the punctured code of $\C$ at the
  position given by $\puncSet$. We propose to decode on $\C'$,
  $\vect{y}'=(1, 1, 1, 1, 1, 1, 1, 1, 1, 0, 0, 0)$ the received word. We
  obtain $\vect{c}_1=(1, 0, 0, 1, 1, 0, 1, 0, 1, 1, 1, 1, 0, 0, 0)$ by
  the decoding of $\vect{y}_1$. We deduce that $\vect{c}_1' =
  \pi_{\bar{\puncSet}}(\vect{c}_1) = (1, 1, 0, 1, 0, 1, 1, 1, 1, 0, 0,
  0)$ is a codeword of $\C'$. Since $d(\vect{c}_1', \vect{y'}) = 2$,
  then only
  $$
  \ceil{\frac{2^{3}}{V_q(3, 1)}} = 2
  $$
  decoding calls is needed. Let $\vect{y}_2 = (1, 1, 1, 1, 1, 1, 1, 1,
  1, 1, 1, 1, 0, 0, 0)$ be the second, and last, virtual received word
  of $\C$. Note that the prefix used in this case, $(1, 1, 1)$, is the
  farthest from the previous one, that is $(0, 0, 0)$. The decoding on
  $\C$ gives us $\vect{c}_2 = (1, 1, 1, 1, 1, 1, 1, 1, 1, 1, 1, 1, 1,
  1, 1)$, which allows us to compute $\vect{c}_2' =
  \pi_{\bar{\puncSet}}(\vect{c}_2) = (1, 1, 1, 1, 1, 1, 1, 1, 1, 1, 1,
  1)$. Since $d(\vect{c}_2', \vect{y}') = 3 > d(\vect{c}_1',
  \vect{y}')$, we deduce that the closest codeword is $\vect{c}_1'$.\\
  Another way to see our trick on this example, is after the
  computation of $\vect{c}_1$, we observe that $d(\vect{c}_1',
  \vect{y}') = 2$, then the closest codeword from $\vect{y}'$ is at
  most at distance 2. Since the decoding algorithm on $\C$ can correct
  3 errors, there is an extra error for the prefix set, see
  Table~\ref{tab:vect3}. Thus the decoding with prefix $(0, 0, 0)$
  manage the prefix vectors of weight 0 and 1. Moreover, the decoding
  with prefix $(1, 1, 1)$ manage the prefix vectors of weight 2 and 3,
  so all prefix vectors are managed only by these 2 decodings.\\
  With this trick, we need only 2 decoding processes whereas our
  previous algorithm needs 8.
\end{example}

\subsection{On the number of punctured positions}

Since the decoding complexity of the punctured code $\C'$ 
is exponential on the number of punctured positions, we
want to minimize this quantity.  To that end  we introduce some new
notation. For a given integer $j\le \rho$,  let 
\begin{eqnarray*}
\Y_j & = & \brace{\vect{y} \in \F_q^n : d(\vect{y},\C)\ge \R -j}\\
\E_j & = & \brace{\vect{y} - \vect{c} \in \F_q^n : \vect{y} \in \Y_j, \vect{c} \in C \mbox{ and }
  d(\vect{y},\vect{c}) = d(\vect{y},\C)}\\
\puncSet_j & = & \bigcap_{\vect{e} \in \E_j} \mbox{supp}(\vect{e}).
\end{eqnarray*}

\begin{proposition}
  \label{prop:CRPuncturedCode}
  Let $j$ be a positive integer and let $\C_j'$ be the code obtained by
  puncturing $\C$ at the positions of $\puncSet_j$. The covering radius of
  $\C_j'$ satisfies
  $$
  \R_j' \le \max \brace{\R - j - 1, \R - |\puncSet_j|}.
  $$
\end{proposition}
\begin{proof}
  Assume there exists $\vect{y}' \in
  \F_q^{n-|\puncSet_j|}$ such that $d(\vect{y}',\C') > \max  \brace{\R - j - 1, \R -
    |\puncSet_j|}$. Let $(\vect{c}',\vect{e}') \in \C' \times
  \F_q^{n-|\puncSet_j|}$ such that
  \mbox{$\vect{y}' = \vect{c}' + \vect{e}'$} and $\mbox{wt}(\vect{e}') =
  d(\vect{y}',\vect{c}') = d(\vect{y}',\C') > \max
  \brace{\R-j-1, \R - |\puncSet_j|}$. Then there exist $\vect{y} \in \Y_j$
  and $ (\vect{c},\vect{e}) \in \C \times \E_j$ such that
  $\pi_{\overline{\puncSet}_j}(\vect{y}) = \vect{y}'$ and
  $$
  \left\lbrace
    \begin{array}{lcl}
       \vect{y} & = &  \vect{c} +  \vect{e},\\
      \pi_{\overline{\puncSet}_j}( \vect{e}) & = &  \vect{e}'.
    \end{array}
  \right.
  $$
  Since $ \vect{e} \in \E_j$, we have $\puncSet_j \subset \mbox{supp} ( \vect{e})$ and so
    $\mbox{wt}( \vect{e})  =  \mbox{wt}( \vect{e}') + | \puncSet_j|> \R$,
  which is impossible by the definition of covering radius. 
\end{proof}

Therefore,  greater $j$ implies smaller $\R-j-1$ but also greater $\R
-|\puncSet_j|$. For steganographic purposes  we must puncture  to obtain a
code $\C'$ with a covering  radius equal to the correction capacity of
the decoding map of $\C$. Since puncturing $n-1$ times leads to a code
of covering radius 0, this is always possible.\\

\begin{algorithm2e}
  \SetKwInOut{Input}{Input}\SetKwInOut{Output}{Output}

  \Input{A code $\C$,  its covering radius  $\R$ and a positive integer $t$.} 
  \Output{A couple $(\C', \puncSet)$, where $\C'$ is obtained from
    $\C$ by puncturing at the positions of  $\puncSet$ and has
    covering radius $t$.}

\BlankLine

$(\C',\R' )\gets (\C, \R)$;\\
$\puncSet \gets \emptyset$;\\
\While {$t \ne \R'$}{
  $CL \gets$ leaders of cosets of $\C'$ with weight greater than $t$;\\
  $i_{\max} \gets$ the position which occurs the most of time in the
  support of $CL$;\\
  $\puncSet \gets \puncSet \cup \brace{i_{\max}}$;\\
  $\C' \gets \C'$ punctured at the position $i_{\max}$;\\
  $\R' \gets$ covering radius of $\C'$;\\
}

\Return{$\parenthesis{\C', \puncSet}$}
  \caption{\label{algo:PuncSet}
    Algorithm to compute the punctured code $\C'$.}
\end{algorithm2e}

The main drawback of this algorithm is that it
could be very expensive in time and memory complexities, since it
requires the computation of the whole set of the coset leaders. Note
however, that this computation must be done just once for code (as for
any other parameter of $\C$).

\subsection{A stegoscheme based on a punctured code}

Let $\C$ be a code of covering radius $\rho$ and $\dec$ a decoding map
correcting $t$ errors. Assume we know the set  $\puncSet$ of positions
to obtain $\C'$,  the punctured code of covering radius $t$, which can
be computed by using Algorithm~\ref{algo:PuncSet}. Moreover let
$\dec'$ be a decoding map for $\C'$ obtained from
Algorithm~\ref{algo:decCPunct}. Finally, let $H'$ be a parity check
matrix of $\C'$. It can be easily obtained from a generator matrix
$G'$ which, in turn, can be easily obtained from a generator matrix
$G$ of $\C$. We have all the ingredients to define a stegoscheme $\S'$
based on the matrix embedding principle with the punctured code
$\C'$. Let $n'$ and $\kpdual$ be, respectively, the length and
redundancy of $\C'$. The embedding  map of $\S'$ is
$$
\begin{array}{rcrcl}
  \Emb: \F_2^{\kpdual} & \times & \F_q^{n'} & \longrightarrow & \F_q^{n'}\\
  (\vect{m}& , & \vect{v}) & \longmapsto & \vect{y} + \dec'(\vect{v} -
  \vect{y}),\\
\end{array}
$$
where $\vect{y}$ is an element of $ \F_q^{n'}$ such that $
 \vect{y}H'^T = \vect{m}$. Recall that when $\C'$ is in systematic
 form, then one can simply take $ \vect{y}= (\vect{0}, \vect{m})$. 
The extracting map is
$$
\begin{array}{rcl}
  \Ext: \F_q^{n'} & \longrightarrow & \F_2^{\kpdual}\\
  \vect{y} & \longmapsto & \vect{y}H'^T.\\
\end{array}
$$
It is simple to check that
$\Ext\parenthesis{\Emb\parenthesis{\vect{m},\vect{v}}} = \vect{m}$ so
that  couple $(\Emb, \Ext)$ defines a true stegoscheme $\S'$.

\subsection{Tradeoff}

In previous sections we have suggested the use of puncturing up to ensuring embedding success.  However, since  complete decoding is known to be  an NP-hard problem
\cite{DecodingProblemGeneral}, in specific situations  this goal may be too ambitious. 

Indeed, we can simply
use the puncturation principle to increase the embedding probability,
but not necessarily up to one. For example, the sender could have a
target embedding probability $p_{\S}<1$, which is not reachable by the
original stegoscheme. Then we can stop the puncturation process when
the target embedding probability is reached. Analogously, since the
complexity of the decoding algorithm for punctured codes is
exponential in the number of the punctured positions, the
steganographer may also limit this number  to a preset maximum,
according to available computing resources. In short, the proposed
method is totally versatile to be modified according to the
requirements of the sender. In the rest of this article, we puncture
up to have an embedding probability $p_{\S} = 1$.

\section{Parameters of the new stegoschemes}
\label{sec:parameter}

Keeping notations as in the previous section, let $\C$ be a code for
which there exists an efficient algorithm capable of correcting $t$
errors. Let $\C'$ be its punctured of covering radius $t$, and let
$\S,\S'$ be the stegoschemes obtained from $\C$ and $\C'$
respectively. The parameters of  $\S'$ are 
$$
a' = \frac{r'}{n'},\; R' = \frac{t}{n'},\; e'= e'_{rel}= \frac{r'}{t},\;
\tilde{R}'=\frac{\tilde{\bound}'}{n'},\; \tilde{e}'=\tilde{e}' _{rel}=
\frac{r'}{\tilde{\bound}'},
$$
where $\tilde{\bound}'$ is the average covering radius of $\C'$. It is
not easy to obtain general closed formulas for these parameters,
since they depend on the number and location of  punctured positions,
which in turn depend on the code $\C$ and the number $t$.
However, except for perfect codes (for which $t$ is just the covering
radius of $\C$ and we do not need to puncture) we have $p_{\S'}=1$,
$n'<n$, $a\ge a'$ and $T'=t<T$, $\tilde{T}'<\tilde{T}$. The embedding
efficiency of $\S'$ may be larger or smaller than the embedding
efficiency of $\S$, for a relative payload fixed.

\subsection{Numerical experiments}

\begin{table*}
  \scriptsize
  \centering
    \begin{tabular}{|c||c|c|c|c|c|c||c|c|c|c|c|c|} \hline
      $m$ & \multicolumn{6}{c||}{BCH$_m(2)$} & \multicolumn{6}{c|}{punctured BCH$_m(2)$} \\ \cline{2-13}
      \rule{0mm}{4mm} & $n$ & $r$ & $a$ & $\tilde{R}$ & $e$ & $\tilde{e}$ & 
      $n'$ & $r'$ & $a'$ & $\tilde{R}'$ & $e'$ & $\tilde{e}'$ \\ \hline\hline
      
      4  & 15  & 8 & 0.533 & 0.164 & 2.67 & 3.25 & 11    & 4  & 0.363  &
      1.375 & 2   & 2.909 \\ \hline
      5  & 31  & 10 & 0.323 & 0.0801 & 3.33 & 4.03 & 28    & 7  & 0.250  &
      1.766 & 3.5 & 3.965 \\ \hline
      6  & 63  & 12 & 0.190 & 0.0396 & 4 & 4.81 & 59    & 8  & 0.135  &
      1.777 & 4    & 4.501 \\ \hline
      7  & 127 & 14 & 0.110 & 0.0197 & 4.66 & 5.61 & 123  & 10 & 0.081 &
      1.879 &  5   & 5.322 \\ \hline
      8  & 255 & 16 & 0.0627 & 0.0098 & 5.34 & 6.41 & 251  & 12 & 0.047 &
      1.939 &  6    & 6.189 \\ \hline
      9  & 511 & 18 & 0.0352 & 0.00489 & 6 & 7.2 & 507  & 14 & 0.027 & 1.969
      &  7    & 7.110 \\ \hline
      10& 1023 & 20 & 0.196 & 0.00244 & 6.66 & 8 & 1018 & 15 & 0.014 & 1.969
      &  7.5 & 7.617 \\ \hline
      11& 2047 & 22 & 0.0107 & 0.00122 & 7.34 & 8.80 & 2042 & 17 & 0.008 &
      1.984 &  8.5 & 8.566 \\ \hline
    \end{tabular}
  \caption{\label{tab:PC2BCH}
    Parameters of stegoschemes arising from binary  $\BCH_m(2)$ codes and punctured  $\BCH_m(2)$ codes.}
\end{table*}

\begin{table*}
  \scriptsize
  \centering
    \begin{tabular}{|c||c|c|c|c|c|c||c|c|c|c|c|c|} \hline
      $m$ & \multicolumn{6}{c||}{BCH$_m(3)$} & \multicolumn{6}{c|}{punctured BCH$_m(3)$} \\ \cline{2-13}
      \rule{0mm}{4mm} & $n$ & $r$ & $a$ & $\tilde{R}$ & $e$ & $\tilde{e}$ & 
      $n'$ & $r'$ & $a'$ & $\tilde{R}'$ & $e'$ & $\tilde{e}'$ \\
      \hline\hline
      4 & 15 & 10 & 0.667 & 0.222 & 2 & 3 & 12 & 7 & 0.583 & 0.191 & 2.33 &
      3.05\\
      \hline
      5 & 31 & 15 & 0.484 & 0.138 & 3 & 3.5 & 25 & 9 & 0.36 & 0.0985 & 3 &
      3.65 \\
      \hline
      6 & 63 & 18 & 0.286 & 0.0645 & 3.6 & 4.43 & 56 & 11 & 0.197 & 0.0434
      & 3.67 & 4.52\\
      \hline
      7 & 127 & 21 & 0.165 & 0.0303 & 4.2 & 5.45 & 121 & 15 & 0.124 &
      0.0230 & 5 & 5.38\\
      \hline
      8 & 255 & 24 & 0.0941 & 0.0151 & 4.8 & 6.25 & 248 & 17 & 0.0686 &
      0.0112 & 5.66 & 6.11\\ 
      \hline
      9 & 511 & 27 & 0.0529 & 0.00751 & 5.4 & 7.03 & 504 & 20 & 0.0397 &
      0.00572 & 6.66 & 6.94\\ 
      \hline
    \end{tabular}
    \caption{\label{tab:PC3BCH}
      Parameters of stegoschemes arising from binary  $\BCH_m(3)$ codes and punctured  $\BCH_m(3)$ codes.}
  \end{table*}

\begin{figure}
  \centering
  \includegraphics[width=8cm]{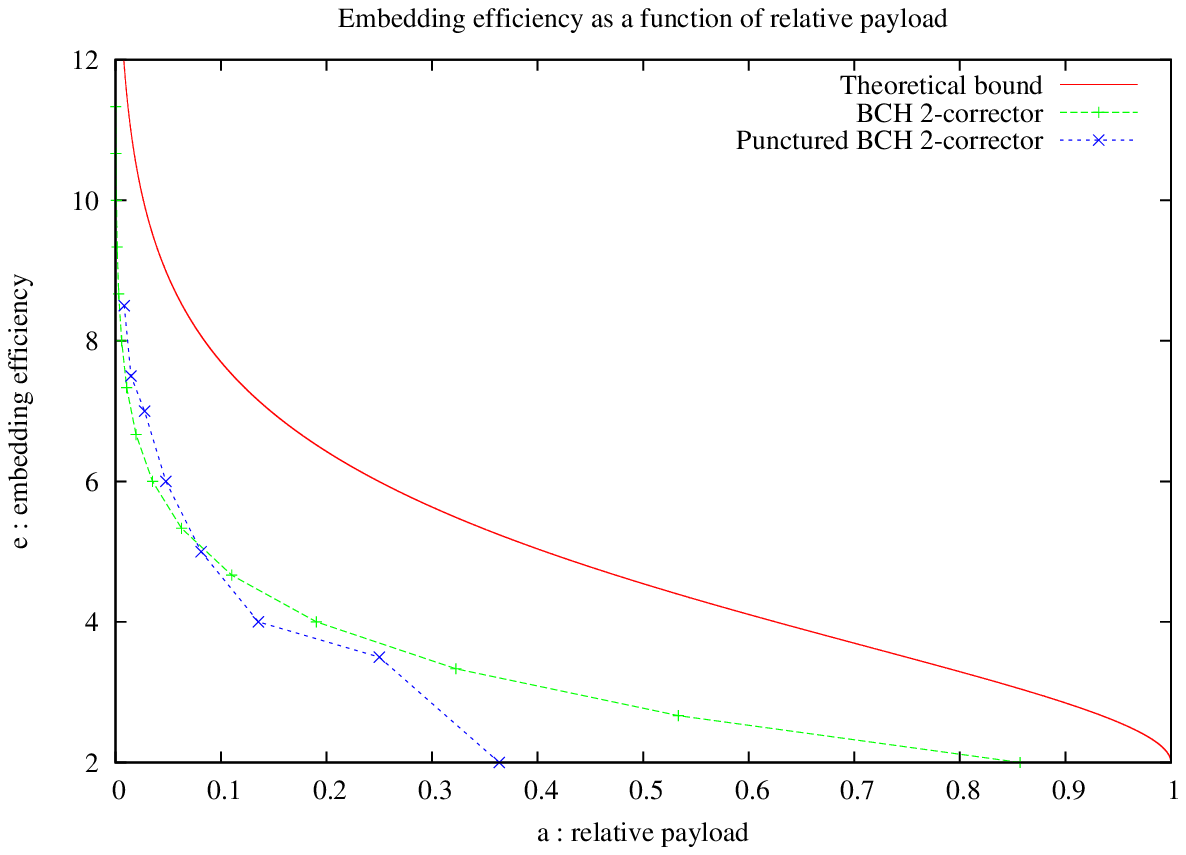}
  \caption{\label{fig:BCH2_VS_PBCH2}
    Comparison between stegoschemes based on the binary  $\BCH_m(2)$ and
    their punctured associated codes.}
\end{figure}

\begin{figure}
  \centering
  \includegraphics[width=8cm]{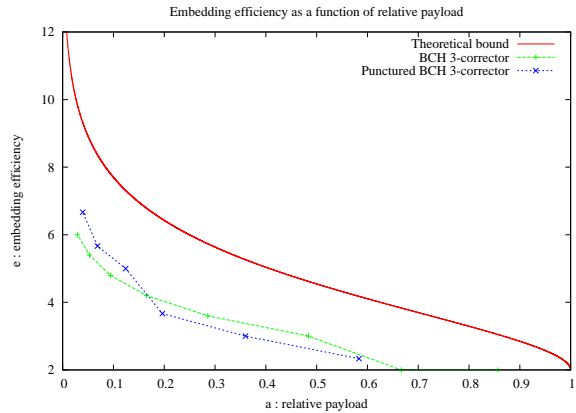}
  \caption{\label{fig:BCH3_VS_PBCH3}
    Comparison between stegoschemes based on the binary  $\BCH_m(3)$ and
    their punctured associated codes.}
\end{figure}

In order to see some concrete results showing the performance of our
method, we list some numerical results obtained for primitive two and
three error correcting binary BCH codes, taking the number $t$ given
by the BCH bound (that is, allowing efficient decoding by means of
Berlekamp-Massey algorithm, as usual). These codes have been studied
by many authors\cite{triple,Helleseth} and proposed as good candidates
for constructing stegoschemes. It is well known that they have
covering radii 3 and 5 respectively, and that $\mbox{BCH}_m(2)$ is
quasi-perfect. By using Algorithm~\ref{algo:PuncSet}, we have
determined  the couple $(\C', \puncSet)$ for some two and three
error-correcting binary BCH codes and subsequently computed the
parameters of the corresponding stegoschemes.  

The obtained results are listed in Tables~\ref{tab:PC2BCH}
(two-error-correcting, $t=2$) and ~\ref{tab:PC3BCH}
(three-error-correcting, $t=3$).  Notations in these tables are same
used throughout this article. In both cases we puncture BCH codes to
get  other codes whose covering radii are equal to the correction
capabilities of BCH codes. Then $R=3, R'=2$ for  two-error-correcting
and $R=5 ,R'=3$ for three-error-correcting BCH codes. We do not
include the probability of embedding success  for punctured codes,
since in all cases such probability is exactly 1. For a better
understanding of these results, and a comparison with stegoschemes
coming from original BCH codes, we  also include a graphical
representation in Figures ~\ref{fig:BCH2_VS_PBCH2} and
~\ref{fig:BCH3_VS_PBCH3}.

As we see, the new stegoschemes may have a better performance in terms
of embedding efficiency. For example,  in Table~\ref{tab:PC2BCH}, we
can see how the stegosystem based on $\BCH_9(2)$ has the same
embedding efficiency as the stegoscheme based on $\BCH_8(2)'$, but
with a smaller relative payload. Thus we conclude that the
stegoschemes based on punctured codes may be better than the original
ones {\em in the worst case}, although they do not appear to be better
in terms of {\em average} embedding efficiency.

All computations have been performed by using the system MAGMA
\cite{magma}, where  BCH codes are managed in systematic form. We
remark that all punctured positions have resulted to be among the
first $n-r$ (systematic) positions. Given our limited computer
resources,  we have not used Algorithm~\ref{algo:PuncSet} in full to
obtain the puncturing of  $\BCH_8(3)$ and $\BCH_9(3)$. Instead  we
have punctured these codes just at the first positions, up to obtain a
code with covering radius $R'$ equal to $t=3$ (and then true results
may be somewhat better than those shown in Table~\ref{tab:PC3BCH}).

\section{Conclusion}

We have proposed a method to ensure or increase the probability of
embedding success for stegoschemes arising from error correcting
codes. This method is based on puncturing codes.  As we have seen, the
use of these punctured codes can also increase the embedding
efficiency of the obtained stegoschemes.

\newcommand{\etalchar}[1]{$^{#1}$}


\end{document}